\documentclass[a4paper,12pt]{article}
\usepackage{amssymb,amsmath,amsthm,mathrsfs}
\newtheorem{theorem}{Theorem}

\newtheorem{proposition}{Proposition}
\newtheorem{lemma}{Lemma}

\newcommand{\F}{\mathbb{F}}
\newcommand{\Tr}{\mathrm{Tr}}

\begin{document}
\title
{\bf Further results on differentially 4-uniform permutations over $\F_{2^{2m}}$}
\author{Zhengbang Zha$^{1,2}$, Lei Hu$^{2,3}$, Siwei Sun$^2$, Jinyong Shan$^2$\\
\vspace*{0.0cm}\\
{\small $^1$School of Mathematical Sciences, Luoyang Normal
University, Luoyang 471022, China}\\
{\small $^2$State Key Laboratory of Information Security, Institute
of Information Engineering, }\\ {\small  Chinese Academy of
Sciences, Beijing 100093, China}\\
{\small $^3$Beijing Center for Mathematics and Information
Interdisciplinary Sciences,}\\
{\small Beijing 100048, China}\\
{\small E-mail: zhazhengbang@163.com; \{hu,swsun,jyshan12\}@is.ac.cn}}
\date{}
\maketitle{}\baselineskip=20pt

\noindent {\small {\bf Abstract} In this paper, we present several new constructions of
differentially 4-uniform permutations over $\F_{2^{2m}}$ by modifying the values of the inverse
function on some subsets of $\F_{2^{2m}}$. The resulted differentially 4-uniform permutations have high
nonlinearities and algebraic degrees, which provide more choices for the design of crytographic substitution boxes.


\vskip 2mm

\noindent {\small {\bf Mathematics Subject Classification (2000)} 94A60; 11T71; 14G50

\section{Introduction}

Many block ciphers use substitution boxes (S-boxes) to
bring the confusion into the cipher, and hence the design of S-boxes
plays an important role in the design of cryptographic systems. To achieve a correct inverse decryption and for ease of the
software implementation, S-boxes are usually designed as permutations over a characteristic 2 finite field of even extension degree, namely $\F_{2^{2m}}$.
To resist against linear cryptanalysis, differential cryptanalysis and other
cryptanalysis like algebraic attacks, one would like the S-boxes having high nonlinearity, low
differential uniformity and high algebraic degree simultaneously.
The inverse function $x^{2^n-2}$ over $\F_{2^n}$ is such a function,
which is used to construct the S-box of the Advanced Encryption Standard (AES) with $n=8$.

The differential uniformity \cite{Nyb} of $f(x)\in\F_{2^n}[x]$ is defined by
$$\Delta_f=\max\{N(a,b)|a,b\in \F_{2^n},a\neq 0\},$$
where $N(a,b)$ denotes the number
of solutions $x\in \F_{2^n}$ of the equation $f(x+a)+f(x)=b$. It is well known
that $\Delta_f=2$ is the minimum possible value of $\Delta_f$. Differentially 2-uniform
functions are called almost perfect nonlinear (APN)
functions which provide the optimal resistance to differential attacks.
A recent progress on APN functions can be found in \cite{Car1,EP}
and the references therein. Up to now, there is only one
known APN permutation, which is defined on $\F_{2^6}$ \cite{BDM}, and the existence of more ones on $\F_{2^{2m}}$ remains open. Therefore,
finding differentially 4-uniform permutations with good cryptographic properties
is an interesting and active research topic for the goal of
providing more choices for S-boxes.

Recently, there has been significant progress in finding functions with low
differential uniformity \cite{Car,CTTL,JZLHH,LW1,PZ,QLDK,QXL,ZW}.
In \cite{BL,BTT} Bracken \emph{et al.} firstly studied highly nonlinear monomials and binomials
which are differentially 4-uniform permutations over $\F_{2^{2m}}$. After that,
many new differentially 4-uniform permutations over $\F_{2^{2m}}$ were constructed
by composing the inverse function and permutations over $\F_{2^{2m}}$.
Qu \emph{et al.} \cite{QTLG,QTTL} applied a powerful switching method (which can be seen in \cite{EP}
and was initially proposed by Dillon in his talk in the ninth international conference on
finite fields and their applications) to construct differentially 4-uniform
permutations and found that the number of CCZ-inequivalent differentially 4-uniform permutations
over $\F_{2^{2m}}$ grows exponentially when $m$ increases.
Later, some new differentially 4-uniform permutations which are
CCZ-inequivalent to the inverse function were obtained by composing the inverse function
and cycles on $\F_{2^{2m}}$ \cite{LWY,YWL}. Zha \emph{et al.} \cite{ZHS} presented two new families of
differentially 4-uniform permutations by modifying the values of the inverse function
on some subfields of $\F_{2^{2m}}$. In \cite{TCT} Tang \emph{et al.} gave a construction providing a large
number of differentially 4-uniform bijections with maximum algebraic degrees and high nonlinearities.
In this paper, we will revisit differentially 4-uniform permutations in \cite{ZHS}
which are constructed by modifying the inverse function on some
subsets of $\F_{2^{2m}}$, and show some new differentially 4-uniform permutations which have
high nonlinearities and algebraic degrees.

The rest of this paper is organized as follows. In Section 2,
some preliminaries needed later as well as a brief overview of known
differentially 4-uniform permutations are presented. We present some new constructions of differentially 4-uniform
permutations in Section 3, and show their cryptographic properties in Section 4.
Finally, we conclude the paper in Section 5.

\section{Preliminaries}\label{prelim}

We define the trace map from $\F_{2^n}$ onto its
subfield $\F_{2^k}$ (with $k|n$) as
$$
\Tr_k^n(x)=x+x^{2^k}+x^{2^{2k}}+\cdots+x^{2^{n-k}}
$$
and denote the absolute trace map from $\F_{2^n}$ onto the binary
subfield $\F_{2}$ by
$\Tr(x)=\sum\limits_{i=0}^{n-1} x^{2^i}$.

The algebraic degree of $f(x)=\sum\limits_{i=0}^{2^n-1} a_ix^i\in
\F_{2^n}[x]$ is denoted by deg $f$, which equals to the maximal 2-weight of the exponent $i$ with
$a_i\neq0$, where the 2-weight of an integer is the number of
ones in its binary expression. It is known that deg $f$ is upper bounded
by $n-1$ if $f$ is a permutation on $\F_{2^n}$. If deg $f\leq 1$, then $f$ is called an affine function.

For a function $f: \F_{2^n}\rightarrow \F_{2^n}$ and any $(a,b)\in\F_{2^n}\times\F_{2^n}^*$, the Walsh
transform of $f$ is defined as
$$
f^\mathcal {W}(a,b):=\sum\limits_{x\in\F_{2^n}}
(-1)^{\Tr(ax+bf(x))}
$$
and the Walsh spectrum of $f$ is $\mathcal {W}_f:=\{f^\mathcal {W}(a,b)|a\in\F_{2^n},
b\in\F_{2^n}^*\}$. The
nonlinearity $\mathcal {NL}(f)$ of $f$ is defined as
$$
\mathcal {NL}(f) \overset{\triangle}{=}
2^{n-1}-\frac{1}{2}\max_{w\in W_f} |w|.
$$
For odd $n$, the nonlinearity $\mathcal {NL}(f)$
is upper bounded by $2^{n-1}-2^{\frac{n-1}{2}}$; and for even $n$
it is conjectured that $\mathcal {NL}(f)$ is upper bounded by
$2^{n-1}-2^{\frac{n}{2}}$ \cite{BTT}. We call a function maximal
nonlinear if its nonlinearity attains these bounds.

Two functions $f,g : \F_{2^n}\to \F_{2^n}$ are called extended
affine equivalent (EA-equivalent) if $g = A_1\circ f \circ A_2 +A$
for some affine permutations $A_1$ and $A_2$ and an affine function
$A$. Nonconstant EA-equivalent functions have the same algebraic
degree.

Two functions $f$ and $g$ from $\F_{2^n}$ to itself are called
Carlet-Charpin-Zinoviev equivalent (CCZ-equivalent) if the graphs of
$f$ and $g$ are affine equivalent. It is shown in \cite{CCZ} that
EA-equivalence implies CCZ-equivalence, but not vice versa.
Every permutation is CCZ-equivalent to its
inverse. Differential spectrum and Walsh spectrum are CCZ-invariants
while algebraic degree is not a CCZ-invariant \cite{BCP, CCZ}.

Many new classes of differentially 4-uniform permutations are introduced
by using the definition of CCZ-equivalence and CCZ-invariants.
One may refer to \cite{BTT, LWY,QTLG,TCT} for recent progress on this topic.

Below we always denote the inverse function by $x^{-1}$ with the
convention that $0^{-1}=0$. Let $n$, $k$ and $q$ be positive integers with $n>1$, $k|n$ and
$q=2^k$. Let $\omega$ be a primitive third root of
unity in the algebraic closure of $\F_q$. Clearly,  $\omega\in\F_q$
if $k$ is even and  $\omega\notin\F_q$ otherwise.

\section{Some constructions of differentially 4-uniform permutations}\label{constructions}

In this section, we revisit a class of differentially 4-uniform permutations
of the form $x^{-1}+t(x^q+x)^{2^n-1}+t$ in \cite{ZHS}. We can not only unify some previous
constructions, but also present new differentially 4-uniform permutations of this form.

Let $S$ be a subset of $\F_{2^n}$ satisfying: (a) either both of 0 and 1 or neither of them belongs to $S$ and;
(b) $\frac{x}{1+x}\in S$ holds for any $x\in S\setminus\{0,1\}$.
Let $\delta_S(x)$ be a characteristic function of $S$,
i.e., $\delta_S(x)=1$ if $x\in S$ and $\delta_S(x)=0$ otherwise.
According to the Lagrange interpolation, we have $\delta_S(x)=1+\prod\limits_{\theta\in S} (x+\theta)^{2^n-1}$.

Inspired by the ideas of \cite{QTTL} and \cite{TCT}, we discuss the cryptographic properties of the function
\begin{equation}\label{1}
f(x)=x^{-1}+\delta_S(x)
\end{equation}
over $\F_{2^n}$ in the sequel.

\begin{proposition} \label{prop1} The function $f$ defined by (1) is a permutation over $\F_{2^n}$ and its compositional
inverse is $g(x)=(x+\delta_{S^\prime}(x))^{-1}$, where $S^\prime$ is a subset of $\F_{2^n}$
satisfying that $x^{-1}+1\in S^\prime$ holds if and only if $x\in S$.
\end{proposition}

\begin{proof} Assume $x,y\in\F_{2^n}$ with $f(x)=f(y)$. Then we have $x^{-1}+\delta_S(x)=y^{-1}+\delta_S(y)$.
If $\delta_S(x)=\delta_S(y)$, we get $x^{-1}=y^{-1}$, which leads to $x=y$.
If $\delta_S(x)\neq\delta_S(y)$, we get $x^{-1}+1=y^{-1}$, which, excluding the cases of $(x,y)=(0,1)$ or $(x,y)=(1,0)$, implies $y=\frac{x}{1+x}$
and $\delta_S(x)\neq\delta_S(\frac{x}{1+x})$. It is a contradiction with the hypothesis on $S$. Then
we deduce that $f$ is a permutation over $\F_{2^n}$. We can easily check that
$g(f(x))=f(g(x))=x$, which completes the proof.
\end{proof}

\begin{proposition} \label{prop2} Let $f$ be defined as in (1). Then
$f$ is a differentially 4-uniform permutation over $\F_{2^n}$ if and only if
for any $a\in\F_{2^n}$ with $a\neq0,1$, one of the following two statements holds:

1) For $\delta_S(a)=\delta_S(0)$, the cases of $\delta_S(\omega a)=\delta_S(\omega^2 a)$
and of $x^2+ax+\frac{a^2}{1+a}=0$ with $\delta_S(x+a)\neq\delta_S(x)$ cannot occur simultaneously; and

2) For $\delta_S(a)\neq\delta_S(0)$, the cases of $\delta_S(\omega a)\neq\delta_S(\omega^2 a)$
and of $x^2+ax+\frac{a^2}{1+a}=0$ with $\delta_S(x+a)=\delta_S(x)$ cannot occur simultaneously.
\end{proposition}

\begin{proof} Let $a,b\in\F_{2^n}$ with $a\neq0$. We consider the solutions of
the equation $f(x+a)+f(x)=b$, i.e.,
\begin{equation}\label{2}
(x+a)^{-1}+\delta_S(x+a)+x^{-1}+\delta_S(x)=b
\end{equation}
over $\F_{2^n}$. By Proposition \ref{prop1} we get that $f$ is a permutation over $\F_{2^n}$,
which implies that Eq. (2) has no solution if $b=0$. Below we assume $b\neq0$ and $x$ is a
solution of (2). When $x=0$ or $a$, we get
$$b=a^{-1}+\delta_S(a)+\delta_S(0)(:=b_0)$$
from (2). For other possible solutions namely ones with $x\neq0$ and $x\neq a$, we
divide into the following two disjoint cases to discuss.

Case I: If $\delta_S(x+a)=\delta_S(x)$, then we have $(x+a)^{-1}+x^{-1}=b$
from (2), which leads to $x^2+ax+\frac{a}{b}=0$.

Case II: If $\delta_S(x+a)\neq\delta_S(x)$, then we have $(x+a)^{-1}+x^{-1}=b+1$
from (2). If $b=1$, we have no solution of (2). Otherwise, we obtain
$x^2+ax+\frac{a}{b+1}=0$ from (2).

For any pair $(a,b)$ with $b\neq b_0$, there are at most four solutions of (2)
in Cases I and II. In the sequel, we consider the case of $b=b_0\neq0$.

If $\delta_S(a)=\delta_S(0)$, then we have $b_0=a^{-1}$. In this case, we
get two solutions $\omega a$ and $\omega^2 a$ in Case I, which need satisfying
$\delta_S(\omega a)=\delta_S(\omega^2 a)$. Similarly in Case II, Eq. (2) turns
to $x^2+ax+\frac{a^2}{1+a}=0$.

If $\delta_S(a)\neq\delta_S(0)$, then we have $b_0=a^{-1}+1$. In this case, we
get two solutions $\omega a$ and $\omega^2 a$ in Case II, which need satisfying
$\delta_S(\omega a)\neq\delta_S(\omega^2 a)$. Similarly in Case I, Eq. (2) becomes $x^2+ax+\frac{a^2}{1+a}=0$.

For constructing differentially 4-uniform permutations,
we want $\Delta_f\leq4$ when $a\neq0,1$ and $b=b_0$.
More precisely, we need to find an appropriate
set $S$ such that there are at most two solutions in Cases I and II when $a\neq0,1$ and $b=b_0$,
which equivalent to the two statements in Proposition \ref{prop2}.

From the definition of $S$, we always get $\delta_S(0)=\delta_S(1)$ and
$\delta_S(\omega)=\delta_S(\omega^2)$. Then we can easily check that
there are exactly four solutions $0,1,\omega,\omega^2$ of (2) when
$a=b_0=1$, which implies that $\Delta_f=4$.
The desired conclusion then follows.
\end{proof}

We note that the function $f(x)=x^{-1}+\delta_s(x)$ is a special case of the switching method
studied in \cite{QTTL}. In Theorem 5.3 of \cite{QTTL}, the composite inverse of $G$ is of the
form $G^{-1}=1/(x+g(x))=1/(x+1)$ if $g(x)=1$; and it is $1/x$ if $g(x)=0$. If $S=\mathrm{supp}(g)$ is taken, then it
immediately lead to the function $f$. Our main contributions are characterizing new subsets $S$
and obtaining new constructions of differentially 4-uniform permutations.

According to Proposition \ref{prop2}, we will give some constructions in the sequel
by choosing different sets $S$.

\subsection{The relevance with known constructions}\label{first}

In this subsection, we will give results related to
some known constructions. First we list a lemma needed below.

\begin{lemma} \label{lem1} \cite{LN}
For any $a,b\in\F_{2^n}$ with $a\neq0$, the polynomial
$f(x)=x^2+ax+b$ is irreducible over $\F_{2^n}$ if and only if
$\Tr(b/a^2)=1$.
\end{lemma}

From Proposition \ref{prop2} and Lemma \ref{lem1}, we can obtain the following construction.
The proof is trivial and we omit it here.

\begin{theorem} \label{thm1} Let $S=\F_{2^k}$. Then $f$ is a differentially 4-uniform permutation
over $\F_{2^n}$ if $k$ is even or $k=1,3$ and $\frac{n}{2}$ is odd.
\end{theorem}

The result of Theorem \ref{thm1} includes the constructions of differentially 4-uniform permutations
in Theorems 1 and 3 of \cite{ZHS} in the case of $t=1$.

We need the following lemma which can be derived from
Lemma 2 in \cite{TCT}. We present here its proof for completeness.

\begin{lemma} \label{lem2} For any $a\in\F_{2^n}/\{0,1\}$, define the
polynomial $\mu(x):=x^2+ax+\frac{a^2}{1+a}\in\F_{2^n}[x]$. If $\mu(x)=0$ has two
solutions $\lambda $ and $\nu$ in $\F_{2^n}$, then we have $\Tr(\frac{1}{\lambda+1})=\Tr(\frac{1}{\nu+1})=0$.
\end{lemma}

\begin{proof} If $\mu(x)=0$ has two solutions $\lambda,\nu$, from Lemma \ref{lem1} we have $\Tr(\frac{1}{1+a})=0$
and $\lambda,\nu\neq0,1$ since $a\neq0,1$.
Since $$\frac{1}{\lambda+1}\cdot\frac{1}{\nu+1}=\frac{1}{\lambda\nu+\lambda+\nu+1}=\frac{1}{a^2(1+a)^{-1}+a+1}=1+a$$
and $\frac{1}{\lambda+1}+\frac{1}{\nu+1}=\frac{\lambda+\nu}{\lambda\nu+\lambda+\nu+1}=a+a^2$, we
get that $\frac{1}{\lambda+1}$ and $\frac{1}{\nu+1}$
are the roots of equation $x^2+(a+a^2)x+1+a=0$.
Note that $\Tr(\frac{1}{\lambda+1})=\Tr(\frac{1}{\nu+1})=0$
if and only if there exist two values $u,v\in\F_{2^n}$ such that $\frac{1}{\lambda+1}=u+u^2$ and $\frac{1}{\nu+1}=v+v^2$.
Assume $u$ and $v$ are the roots of equation $x^2+sx+p=0$, then $u+u^2$ and $v+v^2$ must be the roots
of equation $x^2+(s+s^2)x+p(1+s+p)=0$.

We choose $s=a$. By Lemma \ref{lem1} and $\Tr(\frac{1}{1+a})=0$, there exists an element $p\in\F_{2^n}$ such that
$p^2+(1+a)p+(1+a)=0$ or namely, $p(1+s+p)=1+s$. Then we have $\frac{p^2}{s^2}+\frac{p}{s}+\frac{p}{s^2}+\frac{1}{s}+\frac{1}{s^2}=0$
and $\Tr(\frac{p}{s^2})=0$, which implies that equation $x^2+sx+p=0$ actually has two roots in $\F_{2^n}$.
\end{proof}

\begin{theorem} \label{thm2} Let $S$ be a subset of $\F_{2^n}$ satisfying $\frac{x}{1+x}\in S$
and $\Tr(x)=1$ for any $x\in S$. Then $f$ is a differentially 4-uniform permutation over $\F_{2^n}$.
\end{theorem}

\begin{proof} Since $\Tr(x)=1$ for any $x\in S$, we have $0\not\in S$
and $\Tr(\frac{1}{x+1})=1$ for any $x\in S$. If $\delta_S(a)=\delta_S(0)$,
we can deduce that the solutions of $\mu(x)=0$ satisfying
$\delta_S(x+a)=\delta_S(x)=0$ from Lemma \ref{lem2}. If $\delta_S(a)\neq\delta_S(0)$,
then we get $a\in S$ and $\Tr(\frac{1}{a+1})=1$, which implies that there are
no solutions of $\mu(x)=0$ from Lemma \ref{lem1}. We complete the proof by
Proposition \ref{prop2}.
\end{proof}

We note that the compositional inverse of the permutation $f$ defined in Theorem \ref{thm2} is exactly the known
differentially 4-uniform bijection presented in Construction 1 of \cite{TCT}.

\subsection{New constructions from unions of two subfields of $\F_{2^n}$}\label{first}

In what follows, we introduce two new constructions by combining
two suitable subfields of $\mathbb{F}_{2^n}$. To do this, we need
the following lemma.

\begin{lemma} \label{lem3} Let $k_1$ and $k_2$ be divisors of $n$ and $S=\F_{2^{k_1}}\cup\F_{2^{k_2}}$.
Then $\delta_S(\frac{y^2}{1+y})\neq \delta_S(x^2+xy)$
for any $x,y\in S$ with $x+y\not\in S$.
\end{lemma}

\begin{proof}Assume $x,y\in S$ with $x+y\not\in S$, without loss of generality,
we assume $x\in \F_{2^{k_1}}\setminus\F_{2^{k_2}}$
and $y\in \F_{2^{k_2}}\setminus \F_{2^{k_1}}$. Obviously,
we have $\frac{y^2}{1+y}\in \F_{2^{k_2}}$ and $\delta_S(\frac{y^2}{1+y})=1$.
If $\delta_S(x^2+xy)=1$, then we obtain $x^2+xy\in \F_{2^{k_2}}$, which implies
$$x^2+xy=x^{2^{1+k_2}}+x^{2^{k_2}}y.$$
Raising the above equation by $2^{k_1}$-th powers, we have
$$x^2+xy^{2^{k_1}}=x^{2^{1+k_2}}+x^{2^{k_2}}y^{2^{k_1}}.$$
Then we get
$$(x+x^{2^{k_2}})(y+y^{2^{k_1}})=0$$
by adding the above two equations. It leads to $x\in \F_{2^{k_2}}$ or $y\in \F_{2^{k_1}}$,
which is a contradiction.
\end{proof}

Utilizing Proposition \ref{prop2} and Lemma \ref{lem3}, we have the following theorems.

\begin{theorem} \label{thm3} Let $k_1$ and $k_2$ be even divisors of $n$ and  $S=\F_{2^{k_1}}\cup\F_{2^{k_2}}$.
Then $f$ is a differentially 4-uniform permutation over $\F_{2^n}$.
\end{theorem}

\begin{proof} If $\delta_S(a)=\delta_S(0)$, then we get $a\in S$.
By Lemma \ref{lem3}, we have that $\delta_S(\frac{y^2}{1+y})\neq \delta_S(x^2+xy)$
for any $x,y\in S$ with $x+y\not\in S$, which implies that the equation $\mu(x)=0$ has no solution.

If $\delta_S(a)\neq\delta_S(0)$, then we have $a\not\in S$. We assume
$\delta_S(\omega a)\neq\delta_S(\omega^2 a)$ is true and $\omega a\in S$ and $\omega^2 a\not\in S$
without loss of generality. Since $S=\F_{2^{k_1}}\cup\F_{2^{k_2}}$
and $k_1$ and $k_2$ are even integers, we obtain $\omega^2 a\in S$ from $\omega a\in S$ and $\omega \in S$,
which is a contradiction. The proof is completed.
\end{proof}

\begin{theorem} \label{thm4} Let $k_1$ be an even divisor of $n$ with $\gcd(3,k_1)=1$.
Assume $6|n$, $n/6$ is odd and $S=\F_{2^3}\cup\F_{2^{k_1}}$.
Then $f$ is a differentially 4-uniform permutation over $\F_{2^n}$.
\end{theorem}

\begin{proof} If $\delta_S(a)=\delta_S(0)$, similarly to the proof of Theorem \ref{thm3},
we can show that the equation $\mu(x)=0$ has no solution.

If $\delta_S(a)\neq\delta_S(0)$, $\delta_S(\omega a)\neq\delta_S(\omega^2 a)$ holds
only if $\omega a\in\F_{2^3}$ and $\omega^2 a\not\in S$ or
$\omega^2 a\in\F_{2^3}$ and $\omega a\not\in S$.
Without loss of generality, we consider the case of $\omega a\in\F_{2^3}$ and $\omega^2 a\not\in S$.
Then we have $a^8=\omega^2 a$, $a^{64}=a$ and
$$
\begin{array}{rl}
\Tr(\frac{1}{1+a})&=\Tr^3_1(\Tr^6_3(\Tr^n_6(\frac{1}{1+a})))=\Tr^3_1(\Tr^6_3(\frac{1}{1+a}))\\
&=\Tr^3_1(\frac{1}{1+\omega a+(\omega a)^{-1}})=1
\end{array}
$$
since $\Tr^3_1(\frac{1}{1+x+x^{-1}})=1$ for any $x\in\F_{2^3}$, which implies that
$\mu(x)=0$ has no solution while $\delta_S(\omega a)\neq\delta_S(\omega^2 a)$. This completes the proof.
\end{proof}

\subsection{A construction from unions of subsets of $\F_{2^n}$}\label{second}

In this subsection, we introduce a new construction by combining
some subsets of $\mathbb{F}_{2^n}$.

\begin{theorem} \label{thm5} Let $k$ be even and $\frac{n}{k}$ be odd and let
$l$ be a divisor of $k$. Let $S_1$ be a subset of $\mathbb{F}_{2^n}$ satisfying $\frac{x}{1+x}\in S_1$
and $\Tr(x)=1$ for any $x\in S_1$. Assume $S=S_1\cup (\F_q\setminus\F_{2^l})$ or $S=\F_q\setminus\F_{2^l}$,
then $f$ is a differentially 4-uniform permutation over $\F_{2^n}$ if

1) $l$ is even;

2) $l=1$ and $k\equiv2(\mathrm{mod}\ 4)$; or

3) $l=3$ and $k=6$.
\end{theorem}

\begin{proof} If $\delta_S(a)=\delta_S(0)$, we get $a\not\in S$.
From Lemma \ref{lem2}, if $\mu(x)=0$ has two solutions $\lambda,\nu$, then we have
$\nu=\lambda+a$ and $\lambda,\nu\not\in S_1$. Suppose $\mu(x)=0$ with $\delta_S(\lambda)\neq\delta_S(\nu)$ holds.
Without loss of generality, we assume $\delta_S(\nu)=0$ and $\delta_S(\lambda)=1$, which implies that
$\lambda+a\not\in S$ and $\lambda\in S$. We get $\lambda\in\F_q\setminus\F_{2^l}$ since $\lambda\not\in S_1$.
Then we obtain
\begin{equation}\label{3}
\lambda^2+a \lambda+\frac{a^2}{1+a}=0
\end{equation}
and
\begin{equation}\label{4}
\lambda^2+a^q\lambda+\frac{a^{2q}}{1+a^q}=0.
\end{equation}
Since $a\not\in S$, we have that
$a\not\in\F_q$ or $a\in\F_{2^l}$. If $a\not\in\F_q$, we can deduce that
$\lambda=1+\frac{1}{(1+a)^{q+1}}$, which implies $a^{q^2}=a$. Since $n/k$ is
odd and $a^{q^{n/k}}=a$, we obtain $a^q=a$, which is a contradiction.
If $a\in\F_{2^l}$, then we get $\lambda+a\in\F_q\setminus\F_{2^l}$, which contradicts
our first assumption $\lambda+a\not\in S$. Hence, $\mu(x)=0$ with $\delta_S(x+a)\neq\delta_S(x)$
cannot holds.

If $\delta_S(a)\neq\delta_S(0)$, then we have $a\in S$. When $a\in S_1$, then we
have $\Tr(\frac{1}{1+a})=1$, which implies that there are no solutions of $\mu(x)=0$.
When $a\in\F_q\setminus\F_{2^l}$, we assume that $\delta_S(\omega a)\neq\delta_S(\omega^2 a)$
is true.

If $l$ is even, then we have $\omega a,\omega^2 a\in\F_q\setminus\F_{2^l}$ for $\omega^4=\omega$.
It implies that $\delta_S(\omega a)=\delta_S(\omega^2 a)=1$, which is a contradiction.

If $l=1$ and $k\equiv2 (\mathrm{mod}\ 4)$, $\delta_S(\omega a)\neq\delta_S(\omega^2 a)$ can be true
only if $a=\omega$ or $a=\omega^2$. For $a=\omega,\omega^2$, we obtain that
$\Tr(\frac{1}{1+a})=\Tr_1^2(\frac{1}{1+a})=1$, which implies that $\mu(x)=0$ has no solutions
in this case.

If $k=6$ and $l=3$, $\delta_S(\omega a)\neq\delta_S(\omega^2 a)$ holds
only if $a^8=\omega a$ or $a^8=\omega^2 a$. We consider the case of $a^8=\omega^2 a$
for example. We have $\Tr(\frac{1}{1+a})=\Tr_1^3(\Tr_3^6(\Tr_6^n(\frac{1}{1+a})))=1$, which also
deduce that there are no solutions of $\mu(x)=0$.

Thus, the cases of $\delta_S(\omega a)\neq\delta_S(\omega^2 a)$
and of $\mu(x)=0$ cannot occur simultaneously
if $\delta_S(a)\neq\delta_S(0)$. The proof is finished by Proposition \ref{prop2}.
\end{proof}

\noindent{\bf Remark 1.} {\it In Theorem \ref{thm5}, if $\Tr(x)=1$ for any $x\in \F_q\setminus\F_{2^l}$,
then $f$ is equal to one differential 4-uniform permutation defined in Theorem \ref{thm2}, which is
CCZ-equivalent to the one of Theorem 1 in \cite{TCT}.
Theorem \ref{thm5} exhibits some specific examples of Theorem 5.3 in \cite{QTTL}.}

\subsection{Two constructions from inverse sets of affine subspaces of $\F_{2^n}$}\label{third}

Let $t_1\in\F_{2^n}^*$ with $\Tr_k^n(t_1)=0$. We consider the subset
$S=\{x\in\F_{2^n}: x^{-q}=x^{-1}+t_1\}$, which is the set of the inverses of the elements in the
non-empty affine subspace
$\{x\in\F_{2^n}: x^{q}=x+t_1\}$.
Obviously we have $0\not\in S$ and $\delta_S(0)=0$.

\begin{lemma} \label{lem4} For $a\in S$, the cases of $\delta_S(\omega a)\neq\delta_S(\omega^2 a)$
and of $\mu(x)=0$ cannot occur simultaneously if

1) $q\equiv1(\mathrm{mod}\ 3)$;

2) $q\equiv2(\mathrm{mod}\ 3)$ and $t_1\notin \F_q$; or

3) $\frac{n}{2k}$
is odd, $t_1\in \F_q$ and $q=2$ or 8.
\end{lemma}

\begin{proof} Assume $\delta_S(\omega a)\neq\delta_S(\omega^2 a)$, without loss
of generality, we consider the case of $\omega a\in S$ and $\omega^2 a\not\in S$.
Since $S=\{x\in\F_{2^n}: x^{-q}=x^{-1}+t_1\}$, we get $a^{-q}=a^{-1}+\omega t_1$ when $q\equiv1(\mathrm{mod}\ 3)$
and $a^{-q}=\omega a^{-1}+\omega^2 t_1$ when $q\equiv2(\mathrm{mod}\ 3)$.

Since $a\in S$, we get $a^{-q}=a^{-1}+t_1$. When $q\equiv1(\mathrm{mod}\ 3)$, we have $t_1=0$,
which is a contradiction. When $q\equiv2(\mathrm{mod}\ 3)$, we have $\omega^2a^{-1}+t_1=0$ and $a^{-1}=\omega t_1$,
which implies $t_1^q=t_1$. Therefore, we achieve the goal
when $q\equiv2(\mathrm{mod}\ 3)$ and $t_1^q\neq t_1$. When $\frac{n}{2k}$
is odd and $q=2$ or 8, we need to show that
there are no solutions of $\mu(x)=0$. From Lemma 1, a direct proof is
to show that $\Tr_1^n(\frac{1}{1+a})=1$. We get
$$
\begin{array}{cl}
\Tr_1^n(\frac{1}{1+a})&=\Tr_1^n(\frac{1}{1+(\omega t_1)^{-1}})=\Tr_1^k(\Tr_k^{2k}(\Tr_{2k}^n(\frac{1}{1+\omega t_1})))\\
&=\Tr_1^k(\Tr_k^{2k}(\frac{1}{1+\omega t_1}))=\Tr_1^k(\frac{1}{1+t_1+t_1^{-1}})=1,
\end{array}
$$
which completes the proof.
\end{proof}

\begin{lemma} \label{lem5} If $a\not\in S$, there are no solutions of $\mu(x)=0$
with $\delta_S(x+a)\neq\delta_S(x)$ if the equation
\begin{equation}\label{3}
\begin{array}{rl}
 &(t_1^4+t_1^2)a^{4q}+(t_1^2+t_1)a^{4q-1}+(t_1^2+t_1)a^{4q-2}+a^{4q-3}+a^{4q-4}+(t_1^2+t_1)a^{3q}\\
+&t_1^2a^{3q-1}+(1+t_1)a^{3q-2}+a^{3q-3}+(t_1^2+t_1)a^{2q}+(1+t_1)a^{2q-1}+a^q+a^{q-1}\\
=&1
\end{array}
\end{equation}
has no solution in $\F_{2^n}\setminus \F_q$.
\end{lemma}

\begin{proof} We assume that $\mu(x)=0$ has two solutions $x,x+a$ satisfying $\delta_S(x)\neq\delta_S(x+a)$. Without loss of
generality, we consider the case of $x\in S$ and $x+a\not\in S$. As $x\neq0$, $x\neq a$
and $x^2+ax+\frac{a^2}{1+a}=0$, we have $a\neq1$ and
\begin{equation}\label{4}
x^{-2}+\frac{1+a}{a}x^{-1}+\frac{1+a}{a^2}=0.
\end{equation}
From $x\in S$ we get $x^{-q}=x^{-1}+t_1$. Substituting it to the $q$-th power of (6),
we obtain
\begin{equation}\label{5}
x^{-2}+\frac{1+a^q}{a^q}x^{-1}+t_1^2+\frac{1+a^q}{a^q}t_1+\frac{1+a^q}{a^{2q}}=0.
\end{equation}

If $a\in\F_q$, we get $t_1^2+\frac{1+a}{a}t_1=0$ from (6) and (7), which leads
to $t_1=\frac{1+a}{a}$. Since $(x+a)^2+a(x+a)+\frac{a^2}{1+a}=0$, we derive
$(x+a)^{-2}+\frac{1+a}{a}(x+a)^{-1}+\frac{1+a}{a^2}=0$ and $(x+a)^{-2q}+\frac{1+a}{a}(x+a)^{-q}+\frac{1+a}{a^2}=0$
for $a\in\F_q$. Combining the above two equations, we may draw
the conclusion that $(x+a)^{-1}\in\F_q$ or $x+a\in S$. Since $x+a\not\in S$,
we get $(x+a)^{-1}\in\F_q$ and $x\in\F_q$, which contradicts the first
assumption $x\in S$. Therefore, we obtain $a\not\in\F_q$ and
$$x^{-1}=\frac{(t_1^2+t_1)a^{2q}+a^{2q-1}+a^{2q-2}+(1+t_1)a^q+1}{a^{2q-1}+a^q}$$
from (6) and (7). Substituting it into (6), we can derive Eq. (5) by some trivial
computation. The proof is finished.
\end{proof}

We also need the following lemma.

\begin{lemma} \label{lem6} \cite{LN} An irreducible polynomial over $\F_q$
of degree $n$ remains irreducible over $\F_{q^l}$ if and only if
$\gcd(l,n)=1$.
\end{lemma}

Now we can get many differentially 4-uniform permutations by computing the solutions
of (5). We list two simple examples in the following theorems.

\begin{theorem} \label{thm6} If $q=2$, $\frac{n}{2}$ is odd and $t_1=1$, then $S=\{\omega,\omega^2\}$
and $f$ is a differentially 4-uniform permutation over $\F_{2^n}$.
\end{theorem}

\begin{proof} The case of $a\in S$ is proved by Lemma \ref{lem4}. Now we consider
the case of $a\not\in S$. Since $q=2$, $\frac{n}{2}$ is odd and $t_1=1$, Eq.
(5) becomes $a^4+a^3+a^2+a+1=0$, which has no solution on $\F_{2^n}$. The proof is completed
by Lemma \ref{lem5} and Proposition \ref{prop2}.
\end{proof}

\begin{theorem} \label{thm7} Let $\gcd(n,5)=1$ and $\frac{n}{4}$ be odd.
If $q=4$ and $t_1=1$, then $S=\{x\in\F_{2^n}: x^{-4}=x^{-1}+1\}$ and
$f$ is a differentially 4-uniform permutation over $\F_{2^n}$.
\end{theorem}

\begin{proof} Similarly to the proof of Theorem \ref{thm6}, we need to show that
Eq. (5) has no solution in $\F_{2^n}\setminus \F_q$. Since $q=4$ and $t_1=1$,
Eq. (5) becomes $g(a):=a^{13}+a^{12}+a^{11}+a^9+a^4+a^3+1=0$.
We remark that $g(a)=0$ has no solution in $\F_{2^4}$ and
$$
\begin{array}{rl}
&g(a)=a^{13}+a^{12}+a^{11}+a^9+a^4+a^3+1\\
&=(a^5+a^3+1)(a^4+\omega a^3+\omega a^2+\omega^2 a+\omega^2)(a^4+\omega^2 a^3+\omega^2 a^2+\omega a+\omega)
\end{array}
$$
has only irreducible factors with degrees 2 and 5 over $\F_{2^4}$. Since $\gcd(n,5)=1$
and $\frac{n}{4}$ is odd, then we get that $g(a)=0$ has no solution on $\F_{2^n}$ by
Lemma \ref{lem6}. We complete the proof.
\end{proof}

\section{Cryptographic properties of functions constructed}\label{cryptograhic property}

In this section, we focus on the cryptographic properties of the function $f$ defined by (\ref{1}).
It is shown in \cite{QTTL} that all differentially 4-uniform permutations on $\F_{2^{2m}}$ have
algebraic degree $n-1$. Since $f$ lies in a more general framework in Theorem 5.3 of \cite{QTTL},
we have that $f$ has the maximum possible algebraic degree $n-1$.

In the following, we  give some lower bounds on the nonlinearities of $f$, present some numerical results about the differential spectra and nonlinearities of $f$,
and discuss the CCZ-inequivalence between $f$ and some known differential 4-uniform permutations.

\subsection{Nonlinearity}\label{second}

\begin{lemma} \label{lem7} \cite{LW} Define the Kloosterman sum over $\mathbb{F}_{2^n}$ as
$$K_n(\lambda)=\sum\limits_{x\in \mathbb{F}_{2^n}}(-1)^{{\rm
Tr}(\lambda    x+x^{-1})}, \quad\lambda\in \mathbb{F}_{2^n}.$$
The set $\{K_n(\lambda):\lambda\in \mathbb{F}_{2^n}\}$ is exactly the set of
all integers $t\equiv 0(\mathrm{mod}\ 4)$ in the range
$\left[-2^{\frac{n}{2}+1}+1,2^{\frac{n}{2}+1}+1\right]$.
\end{lemma}

Let $\lfloor 2^{\frac{1}{2}+1}\rfloor=2$ and $\lfloor 2^{\frac{k}{2}+1}\rfloor=t$ for $k>1$,
$t\equiv 0 (\mathrm{mod}\ 4)$ and $2^{\frac{k}{2}+1}-4< t\leq2^{\frac{k}{2}+1}$. We have the following
result by Lemma \ref{lem7}.

\begin{lemma} \label{lem8} Let $a,b\in\F_{2^n}$ with $\Tr(b)=1$. Let $k$ be a
positive integer with $k\mid n$. Then $|\sum\limits_{x\in \F_{2^k}} (-1)^{\Tr(ax+b x^{-1})}|\leq\lfloor 2^{\frac{k}{2}+1}\rfloor$.
\end{lemma}

\begin{proof} Since $\Tr(b)=1$, for any $x\in\F_{2^k}$ we get $\Tr_k^n(b)\neq0$ and
$$\Tr(ax+b x^{-1})=\Tr_1^k(\Tr_k^n(ax+b x^{-1}))=\Tr_1^k(x\Tr_k^n(a)+x^{-1}\Tr_k^n(b)).$$
By Lemma \ref{lem7}, we have that
$$|\sum\limits_{x\in \F_{2^k}} (-1)^{\Tr(ax+b x^{-1})}|
=|\sum\limits_{x\in \F_{2^k}} (-1)^{\Tr_1^k(x\Tr_k^n(a)+x^{-1}\Tr_k^n(b))}|\leq \lfloor2^{\frac{k}{2}+1}\rfloor.$$
\end{proof}

\begin{lemma} \label{lem9} \cite{TCT} For any even $n\geq6$,
we have $\mathcal {NL}(f)\geq 2^{n-1}-2^{n/2}-\mid S\mid$.
\end{lemma}

From Lemma \ref{lem9}, we get a direct lower bound on the nonlinearities of $f$ defined
in Theorems \ref{thm3}-\ref{thm7}. For example, we obtain $\mathcal {NL}(f)\geq 2^{n-1}-2^{n/2}-2$ for the function defined
in Theorem \ref{thm6} and $\mathcal {NL}(f)\geq 2^{n-1}-2^{n/2}-4$ for the function defined in Theorem \ref{thm7}.
In the following, we give a more tightly bound on the nonlinearities of
$f$ defined in Theorems \ref{thm3} and \ref{thm4}.

\begin{proposition} \label{pro4} If $S=\F_{2^{k_1}}\cup\F_{2^{k_2}}$,
the nonlinearity of the function $f$ defined by (\ref{1})
satisfies $$\mathcal {NL}(f)\geq 2^{n-1}-\lfloor2^{\frac{n}{2}}\rfloor-
\lfloor2^{\frac{k_1}{2}+1}\rfloor-\lfloor2^{\frac{k_2}{2}+1}\rfloor-\lfloor2^{\frac{\gcd(k_1,k_2)}{2}+1}\rfloor.$$
Especially, we have $$\mathcal {NL}(f)\geq
2^{n-1}-\lfloor2^{\frac{n}{2}}\rfloor-\lfloor2^{\frac{k_2}{2}+1}\rfloor$$ if $k_1\mid k_2$ and $$\mathcal {NL}(f)\geq
2^{n-1}-\lfloor2^{\frac{n}{2}}\rfloor-\lfloor2^{\frac{k_2}{2}+1}\rfloor-6$$
if $k_1=3$ and $\gcd(k_2,3)=1$.
\end{proposition}

\begin{proof} As we know, the nonlinearity of $f$ is
defined by $\mathcal {NL}(f)=2^{n-1}-\frac{1}{2}\max f^\mathcal
{W}(a,b)$, where $f^\mathcal {W}(a,b)=\sum\limits_{x\in\F_{2^n}}
(-1)^{\Tr(ax+bf(x))}$ for $a,b\in\F_{2^n}$ and $b\neq0$. We have
that
$$
\begin{array}{cl}
f^\mathcal {W}(a,b)&=\sum\limits_{x\in\F_{2^n}} (-1)^{\Tr(ax+b(x^{-1}+\delta_S(x)))}\\
&=\sum\limits_{x\in\F_{2^n}\setminus S} (-1)^{\Tr(ax+b
x^{-1})}+\sum\limits_{x\in S} (-1)^{\Tr(ax+bx^{-1}+b)}.
\end{array}
$$
If $\Tr(b)=0$, we get $|f^\mathcal
{W}(a,b)|=|\sum\limits_{x\in\F_{2^n}} (-1)^{\Tr(ax+b x^{-1})}|\leq
\lfloor2^{\frac{n}{2}+1}\rfloor$ by Lemma \ref{lem7}. If $\Tr(b)=1$, we obtain
$$
\begin{array}{cl}
|f^\mathcal {W}(a,b)|&=|\sum\limits_{x\in\F_{2^n}} (-1)^{\Tr(ax+b x^{-1})}-2\sum\limits_{x\in S} (-1)^{\Tr(ax+b x^{-1})}|\\
&\leq |\sum\limits_{x\in\F_{2^n}} (-1)^{\Tr(ax+b x^{-1})}|+2|\sum\limits_{x\in S} (-1)^{\Tr(ax+b x^{-1})}|.
\end{array}
$$
Notice that
$$
\begin{array}{cl}
\sum\limits_{x\in S} (-1)^{\Tr(ax+b x^{-1})}=&\sum\limits_{x\in \F_{2^{k_1}}} (-1)^{\Tr(ax+b x^{-1})}
+\sum\limits_{x\in \F_{2^{k_2}}} (-1)^{\Tr(ax+b x^{-1})}\\
&-\sum\limits_{x\in \F_{2^{\gcd(k_1,k_2)}}} (-1)^{\Tr(ax+b x^{-1})},
\end{array}
$$
then we obtain $$|f^\mathcal{W}(a,b)|\leq \lfloor2^{\frac{n}{2}+1}\rfloor+2(\lfloor2^{\frac{k_1}{2}+1}\rfloor
+\lfloor2^{\frac{k_2}{2}+1}\rfloor+\lfloor2^{\frac{\gcd(k_1,k_2)}{2}+1}\rfloor)$$
from Lemma \ref{lem8}, which implies that $$\mathcal {NL}(f)\geq 2^{n-1}-\lfloor2^{\frac{n}{2}}\rfloor-\lfloor2^{\frac{k_1}{2}+1}
\rfloor-\lfloor2^{\frac{k_2}{2}+1}\rfloor-\lfloor2^{\frac{\gcd(k_1,k_2)}{2}+1}\rfloor.$$
If $k_1\mid k_2$, we can get $S=\F_{2^{k_2}}$ and $\mathcal {NL}(f)\geq
2^{n-1}-\lfloor2^{\frac{n}{2}}\rfloor-\lfloor2^{\frac{k_2}{2}+1}\rfloor$ by Lemma \ref{lem8} directly.
Furthermore, if $k_1=3$ and $k_2$ is even, we have
$$\mathcal {NL}(f)\geq 2^{n-1}-\lfloor2^{\frac{n}{2}}\rfloor-\lfloor2^{\frac{k_2}{2}+1}\rfloor-6$$
since $|\sum\limits_{x\in \F_{2^3}} (-1)^{\Tr(ax+b x^{-1})}|\leq4$ and $|\sum\limits_{x\in \F_{2}} (-1)^{\Tr(ax+b x^{-1})}|\leq2$.
\end{proof}

\subsection{Numerical results and CCZ-inequivalence}\label{third}

For even $n\leq 12$, we computed the nonlinearities and differential spectra of the functions in Section 3 by using the MAGMA software system.
Some of the computational results are listed in Tables 1-3, where the notation $\mathcal{D}(f)$
represents the differential spectrum of a function $f$, $\mathcal{B}(f)$ represents a bound on the nonlinearity of $f$, and the multiset
$M=\{a_1^{m_1},a_2^{m_2},\cdots,a_{t}^{m_{t}}\}$ means the elements
$a_i$ appears $m_i$ times in $M$ for $1\leq i\leq t$.

\begin{table*}[!h1]
\begin{center}
\caption{Nonlinearities and differential spectra of functions in
Section 3 over $\F_{2^6}$}
\begin{tabular}{|c|c|c|c|c|c|c}\hline
$S$ &  $f$             & $\mathcal{NL}(f)$ &
 $\mathcal{D}(f)$ & $\mathcal{B}(f)$

  \\\hline
  $\F_2$ & Theorem \ref{thm1}    & 24 & $\{0^{2079},2^{1890},4^{63}\}$ & 22

 \\\hline
  $\F_{2^2}$ & Theorem \ref{thm1}    & 22 & $\{0^{2127},2^{1794},4^{111}\}$ & 20

\\\hline
  $\F_{2^3}$ & Theorem \ref{thm1}    & 22 & $\{0^{2199},2^{1650},4^{183}\}$ & 20

\\\hline
  $\F_{2^2}\cup\F_{2^3}$ & Theorem \ref{thm4} & 20 & $\{0^{2247},2^{1554},4^{231}\}$ & 14

\\\hline
  $\F_{2^2}\setminus\F_{2}$ & Theorems \ref{thm5}, \ref{thm6} & 22 & $\{0^{2127},2^{1794},4^{111}\}$ & 22

\\\hline

\end{tabular}
\end{center}
\end{table*}

\begin{table*}[!h2]
\begin{center}
\caption{Nonlinearities and differential spectra of functions in
Section 3 over $\F_{2^{10}}$}
\begin{tabular}{|c|c|c|c|c|c|c}\hline
$S$ &  $f$              & $\mathcal{NL}(f)$ &
 $\mathcal{D}(f)$ & $\mathcal{B}(f)$

  \\\hline
  $\F_{2}$ & Theorem \ref{thm1}    & 480 & $\{0^{524799},2^{521730},4^{1023}\}$ & 478

 \\\hline
  $\F_{2^2}$ & Theorem \ref{thm1}    & 478 & $\{0^{34335},2^{29250},4^{1695}\}$ & 476

  \\\hline
  $\F_{2^2}\setminus\F_{2}$ & Theorems \ref{thm5}, \ref{thm6}   & 478 & $\{0^{525879},2^{519570},4^{2103}\}$ & 478

\\\hline

\end{tabular}
\end{center}
\end{table*}

\begin{table*}[!h3]
\begin{center}
\caption{Nonlinearities and differential spectra of functions in
Section 3 over $\F_{2^{12}}$}
\begin{tabular}{|c|c|c|c|c|c|c}\hline
$S$ &  $f$             & $\mathcal{NL}(f)$ &
 $\mathcal{D}(f)$ & $\mathcal{B}(f)$

  \\\hline
  $\F_{2^2}$ & Theorem \ref{thm1}    & 1982 & $\{0^{8394735},2^{8370210},4^{8175}\}$ & 1980

 \\\hline
  $\F_{2^4}$ & Theorem \ref{thm1}    & 1978 & $\{0^{8419263},2^{8321154},4^{32703}\}$ & 1968

\\\hline
  $\F_{2^6}$ & Theorem \ref{thm1}    & 1970 & $\{0^{8511615},2^{8136450},4^{125055}\}$ & 1920

\\\hline
  $\F_{2^4}\cup\F_{2^6}$ & Theorem \ref{thm3} & 1966 & $\{0^{8534127},2^{8091426},4^{147567}\}$ & 1908

\\\hline
  $\F_{2^4}\setminus\F_{2^2}$ & Theorem \ref{thm5} & 1978 & $\{0^{8415183},2^{8329314},4^{28623}\}$ & 1972

\\\hline
  $x^{-4}+x^{-1}=1$ & Theorem \ref{thm7} & 1980 & $\{0^{8399055},2^{8361570},4^{12495}\}$ & 1980

\\\hline

\end{tabular}
\end{center}
\end{table*}

By Tables 1-3, we conclude that Theorems \ref{thm3}-\ref{thm7} present several differentially
4-uniform permutations CCZ-inequivalent to the ones defined in Theorem 1. Especially,
compared to the numerical results in Tables III and IV of \cite{QTTL}, Tables 2 and 3 list
some new results on nonlinearities and differential spectra when $n=10$ and 12.

For $n=12$, by a Magma computation, we checked that there are 1036 elements of $\F_{2^{12}}$ satisfying
$\Tr(x)=\Tr(\frac{x}{1+x})=1$. That is to say, there are $2^{518}-1$ different sets $S_1$
and $2^{518}$ different sets $S$ in Theorem \ref{thm5}.
For a random choice set $S_1$, the functions constructed in Theorem \ref{thm5} ($S=S_1\cup (\F_{16}\setminus\F_4)$)
and Theorem \ref{thm2} ($S=S_1$) often have different differential spectra
or nonlinearities, which implies that they are CCZ-inequivalent.
We randomly choose 10000 different sets $S_1$ and compute the nonlinearities of
the differentially 4-uniform permutations defined in Theorems \ref{thm5} and \ref{thm2} separately.
The computational results are listed in Table 4. The notations Ave$(\mathcal{NL}(f))$,
Max$(\mathcal{NL}(f))$ and Min$(\mathcal{NL}(f))$ denote the average nonlinearity, maximal nonlinearity
and minimal nonlinearity of $f$ respectively.

\begin{table*}[!h4]
\begin{center}
\caption{Variance of nonlinearities of 10000 samples $f$ on $\F_{2^{12}}$}
\begin{tabular}{|c|c|c|c|c|c|c}\hline
$S$ &  $f$             & Ave$(\mathcal{NL}(f))$ &
 Max$(\mathcal{NL}(f))$ & Min$(\mathcal{NL}(f))$

  \\\hline
  $S_1$ & Theorem \ref{thm2}  & 1911.106 & 1982 & 1864

 \\\hline
  $S_1\cup (\F_{16}\setminus\F_4)$ & Theorem \ref{thm5} & 1910.264 & 1978 & 1866

\\\hline

\end{tabular}
\end{center}
\end{table*}

From Table 4, we found that the average nonlinearity of the functions $f$ in Theorem \ref{thm5}
is less than that of the functions $f$ in Theorem \ref{thm2}. This is due to
$$|S_1\cup (\F_{16}\setminus\F_4)|\geq |S_1|+4$$
since there are exactly four solutions of $\Tr(x)=0$ with $x\in (\F_{16}\setminus\F_4)$ over $\F_{2^{12}}$.
It seems that the nonlinearity of $f$ gets smaller when $|S|$ gets larger, this, however, is not always true, below we give
some examples.

\vskip 2mm

\noindent{\bf Example 1.} {\it For $n=12$, let $\alpha$ be a primitive element of $\F_{2^{12}}$ defined by
$\alpha^{12}+\alpha^7+\alpha^4+\alpha^3+1=0$ and $S_1=\{\alpha^i| i\in \Pi\}$, where

\noindent $\Pi=\{649, 3411, 2016, 2422, 437, 903, 3963, 2464, 1914, 1180, 3755, 2410, 119, 647, 3624, 841, \\
\hskip 3mm 2833, 2709, 2352, 4092, 3812, 2696, 3166, 2950, 1784, 3475, 233, 1831, 1157, 1422, 3897, 2429, \\
\hskip 3mm 918, 1363, 2910, 3955, 1372, 1785, 3013, 1589, 2021, 1721\}.$

\noindent We can easily check that $|S_1|=42$. The nonlinearity of the function $f$ in Theorem \ref{thm5} is 1958, while the function in Theorem \ref{thm2} is of nonlinearity 1956.}

\vskip 2mm

\noindent{\bf Example 2.} {\it For $n=12$, let $\alpha$ be a primitive element of $\F_{2^{12}}$ defined in Example 1
and $S_1=\{\alpha^i| i\in \{3351, 1475, 777, 661, 921, 977, 4076, 2037, 3359, 2414, 3616, 3033, 3401, 3697,\\
3459, 654, 3160, 123, 3226, 2837, 526, 2832, 1182, 4094, 3964, 3887, 1705, 2489, 1766, 4066, 589,\\
184, 1842,2752\}\}.$
We can easily check that $|S_1|=34$ and the nonlinearities of $f$ defined in Theorems \ref{thm5} and \ref{thm2}
are both 1962.}

\vskip 2mm

\noindent{\bf Example 3.} {\it For $n=12$, let $\alpha$ be a primitive element of $\F_{2^{12}}$ defined in Example 1
and $$S_1=\{\alpha^i\mid i\in \{273,546,1092,2184\}\}.$$
We can check that $S_1$ is a subset of $\F_{16}\setminus\F_4$ and $S=\F_{16}\setminus\F_4$.
The nonlinearity of the function $f$ in Theorem \ref{thm5} is 1978, while the function in Theorem \ref{thm2}
is of nonlinearity 1980.}

\section{Conclusion}\label{conclusion}

In this paper, we refine a general technique for constructing a class of differentially
4-uniform permutations by modifying the values of the inverse function on a subset $S$ of $\F_{2^{2m}}$.
By using this technique, we get many differentially 4-uniform permutations with
high nonlinearities and algebraic degrees. Our numerical results support that some of them
are new and have the nonlinearity approaching the maximum while the size of $S$ is not too small.

\vskip 5mm

\noindent {\bf Acknowledgement: } The authors would like to thank
the anonymous reviewers for their valuable comments and suggestions which improved both the quality
and presentation of this paper. The work of this paper was supported by the National Basic Research Programme under Grant 2013CB834203,
the National Natural Science Foundation of China (Grants 11201214 and 61472417),
and the Strategic Priority Research Program of Chinese Academy of Sciences under Grant XDA06010702.

\end{document}